\documentclass[12pt]{article}

\newcommand{\bbsm}{ \begin{smallmatrix}}
	\newcommand{\besm}{\end{smallmatrix}}

\usepackage[hmargin=1.3in,vmargin=1.3in]{geometry}
\usepackage[hmargin=1.3in,vmargin=1.3in]{geometry}
\usepackage{bbding}
\usepackage{mathrsfs}
\usepackage{cite}
\usepackage{amsfonts}
\usepackage{multirow}
\usepackage{amsfonts,amssymb,amsmath,amsthm,bm}
\usepackage[colorlinks,
            linkcolor=blue,
            anchorcolor=blue,
            citecolor=blue]{hyperref}

\newtheorem{theorem}{Theorem}
\newtheorem{example}{Example}
\newtheorem{definition}{Definition}

\newtheorem{lemma}{Lemma}
\newtheorem{proposition}{Proposition}

\begin{document}

\title{Nonexistence of perfect permutation codes under the Kendall $\tau$-metric}
\author{\small Xiang Wang $^{1}$ \ Wang Yuanjie $^{1}$ \ Yin Wenjuan$^{2}$ \footnote{Corresponding Author}\ Fang-Wei Fu$^{2}$ \\
\scriptsize $^1$ National Computer Network Emergency Response Technical Team/Coordination Center of China (CNCERT/CC) \\
\scriptsize $^2$ Chern Institute of Mathematics and LPMC, Nankai University, Tianjin 300071, China\\
\scriptsize E-mail: xqwang@mail.nankai.edu.cn, wang.yuanjie@outlook.com, ywjaimama@163.com, fwfu@nankai.edu.cn\\
}
\date{}
\maketitle
\thispagestyle{empty}
\begin{abstract}
In the rank modulation scheme for flash memories, permutation codes have been studied. In this paper, we study perfect permutation codes in $S_n$, the set of all permutations on $n$ elements, under the Kendall $\tau$-Metric. We answer one open problem proposed by Buzaglo and Etzion. That is, proving the nonexistence of perfect codes in $S_n$, under the Kendall $\tau$-metric, for more values of $n$. Specifically, we present the recursive formulas for the size of a ball with radius $r$ in $S_n$ under the Kendall $\tau$-metric. Further, We prove that there are no perfect $t$-error-correcting codes in $S_n$ under the Kendall $\tau$-metric for some $n$ and $t=2,3,4,~\text{or}~5$.

\end{abstract}


\maketitle

\section{Introduction}
Flash memory is a non-volatile storage medium that is both electrically programmable and erasable. The rank modulation scheme for flash memories has been proposed in \cite{Jiang1}. In this scheme, one permutation corresponds to a relative ranking of all the flash memory cells' levels. A permutation code is a nonempty subset of $S_n$, where $S_n$ is the set of all the permutations over $\{1,2,...,n\}$.
Permutation codes have been studied under various metrics, such as the $\ell_{\infty}$-metric \cite{Klve,Wang2,Yehezkeally}, the Ulam metric \cite{Farnoud}, and the Kendall $\tau$-metric \cite{Jiang2,Wang,Barg,Zhang}.

In this paper, we will focus on  permutation codes under the Kendall $\tau$-metric. The \emph{Kendall $\tau$-distance} \cite{Yehezkeally} between two permutations $\pi, \sigma\in S_n$ is the minimum number of adjacent transpositions required to obtain the permutation $\sigma$ from $\pi$, where an adjacent transposition is an exchange of two distinct adjacent elements. Permutation codes under the Kendall $\tau$-distance with minimum distance $d$ can correct up to $\lfloor\frac{d-1}{2} \rfloor$ errors. Let $A(n,d)$ be the size of the largest code in $S_n$ with minimum Kendall $\tau$-distance $d$. The bounds on $A(n,d)$ were proposed in \cite{Jiang2,Buzaglo,Wang3,Vijayakumaran}. Some $t$-error-correcting codes in $S_n$ were constructed in \cite{Horvitz,Jiang2,Barg,Zhou1,Zhou2}. Buzaglo and Etzion \cite{Buzaglo} proved that there are no perfect single-error-correcting codes in $S_n$, where $n>4$ is a prime or $4\leq n\leq 10$. They further \cite{Buzaglo} proposed the open
problem to prove the nonexistence of perfect codes in $S_n$, under the Kendall $\tau$-metric, for more values of $n$ and/or other distances. In this paper, we prove that there are no perfect $t$-error-correcting codes in $S_n$ under the Kendall $\tau$-metric for some $n$ and $t=2,3,4,~\text{or}~5$. Specially, we prove that there are no perfect two-error-correcting codes in $S_n$, where $n+2>6$ is a prime. We also prove that there are no perfect three-error-correcting codes in $S_n$, where $n+1>6$ is a prime, $n^2+2n-6$ has a prime factor $p>n$, or $4\leq n\leq 33$. We further prove that there are no perfect four-error-correcting codes in $S_n$, where $n+1>6$ or $n+2>7$ is a prime, $n^2+3n-12$ has a prime factor $p>n$, or $5\leq n\leq 19$. Finally, we prove that there are no perfect five-error-correcting codes in $S_n$, where $n+7\geq 12$ is a prime or $n^3+3n^2-6n-28$ has a prime factor $p>n$.

The rest of this paper is organized as follows. In Section \ref{sec2}, we will give some basic definitions for the Kendall $\tau$-metric and for perfect permutation codes. In Section \ref{sec3}, we determine the size of some balls with radius $r$ in $S_n$ under the Kendall $\tau$-metric. In Section \ref{sec4}, we prove the nonexistence of a perfect $t$-error-correcting code in $S_n$ for some $n$ and $t=2,3,4,~\text{or}~5$ by using the sphere packing upper bound. Section \ref{sec5} concludes this paper.

\section{Preliminaries}
\label{sec2}
In this section we give some definitions and notations for the Kendall $\tau$-metric and perfect permutation codes. In addition, we summarize some important known facts.

Let $[n]$ denote the set $\{1,2,...,n\}$. Let $S_n$ be the set of all the permutations over $[n]$. We denote by $\pi\triangleq[\pi(1),\pi(2),...,\pi(n)]$ a \emph{permutation} over $[n]$.  For two permutations $\sigma,\pi \in S_n$, their multiplication $\pi\circ\sigma$ is denoted by the composition of $\sigma$ on $\pi$, i.e., $\pi\circ\sigma(i)=\sigma(\pi(i))$, for all $i\in [n]$. Under this operation, $S_n$ is a noncommutative \emph{group} of size $|S_n|=n!$. Denote by $\epsilon_n\triangleq[1,2,...,n]$ the identity permutation of $S_n$. Let $\pi^{-1}$ be the \emph{inverse} element of $\pi$, for any $\pi\in S_n$.  For an unordered pair of distinct numbers $i,j\in[n]$, this pair forms an inversion in a permutation $\pi$ if $i<j$ and simultaneously $\pi(i)>\pi(j)$.

Given a permutation $\pi=[\pi(1),\pi(2),...,\pi(i),\pi(i+1),...\pi(n)]\in S_n$, an adjacent transposition is an exchange of two adjacent elements $\pi(i),\pi(i+1)$, resulting in the permutation $[\pi(1),\pi(2),...,\pi(i+1),\pi(i),...\pi(n)]$ for some $1\leq i\leq n-1$. For any two permutations $\sigma,\pi\in S_n$, the Kendall $\tau$-distance between two permutations $\pi, \sigma$, denoted by $d_K(\pi,\sigma)$, is the minimum number of adjacent transpositions required to obtain the permutation $\sigma$ from $\pi$. The expression for $d_K(\pi,\sigma)$ \cite{Jiang2} is as follows:
\begin{equation}
d_{K}(\sigma,\pi)=|\{(i,j):\sigma^{-1}(i)<\sigma^{-1}(j)\wedge\pi^{-1}(i)>\pi^{-1}(j)\}|\nonumber.
\end{equation}

For $\pi\in S_n$, the Kendall $\tau$-weight of $\pi$, denoted by $w_K(\pi)$, is defined as the Kendall $\tau$-distance between $\pi$ and  the identity permutation $\epsilon_n$. Clearly, $w_K(\pi)$ is the number of inversions in the permutation $\pi$.
\begin{definition}
For $1\leq d \leq \binom{n}{2}$, $C\subset S_n$ is an $(n,d)$-\emph{permutation code under the Kendall $\tau$-metric}, if $d_{K}(\sigma,\pi)\geq d$ for any two distinct permutations $\pi, \sigma\in C$.
\end{definition}

For a permutation $\pi\in S_n$, the Kendall $\tau$-ball of radius $r$ centered at $\pi$, denoted as $B_{K}^{n}(\pi,r)$, is defined by $B_{K}^n(\pi,r)\triangleq\{\sigma\in S_n|d_K(\sigma,\pi)\leq r\}$. For a permutation $\pi\in S_n$, the Kendall $\tau$-sphere of radius $r$ centered at $\pi$, denoted as $S_{K}^{n}(\pi,r)$, is defined by $S_{K}^n(\pi,r)\triangleq\{\sigma\in S_n|d_K(\sigma,\pi)= r\}$. The size of a Kendall $\tau$-ball or a $\tau$-sphere of radius $r$ does not depend on the center of the ball under the Kendall $\tau$-metric. Thus, we denote the size of $B_{K}^{n}(\pi,r)$ and $S_{K}^{n}(\pi,r)$ as $B_K^n(r)$ and $S_K^n(r)$, respectively. We denote the largest size of an $(n, d)$-permutation code under the Kendall $\tau$-metric as $A_K(n,d)$. The sphere-packing bound for permutation codes under the Kendall $\tau$-metric are as follows:

\begin{proposition}\cite[Theorems 17 and 18]{Jiang2}
\begin{equation}
A_{K}(n,d)\leq\frac{n!}{B_K^n(\lfloor \frac{d-1}{2}\rfloor)}\nonumber.
\end{equation}
\end{proposition}

When $d=2r+1$, an $(n,2r+1)$-permutation code $C$ under the Kendall $\tau$-metric is called a perfect permutation code under the Kendall $\tau$-metric if it attains the sphere-packing bound, i.e., $|C|\cdot B_K^n(r)=n!$. That is,  the balls with radius $r$ centered at the codewords of $C$ form a partition of $S_n$. A perfect $(n,2r+1)$-permutation code under the Kendall $\tau$-metric is also called a perfect $r$-error-correcting code under the Kendall $\tau$-metric.

In \cite{Buzaglo}, Buzaglo and Etzion proved that there does not exist a perfect one-error-correcting code under the Kendall $\tau$-metric if $n>4$ is a prime or $4\leq n\leq 10$. Based on the above definitions and notations, we will prove the nonexistence of a perfect $t$-error-correcting code in $S_n$ under the Kendall $\tau$-metric for some $n$ and $t=2,3,4,~\text{or}~5$ by using the sphere-packing upper bound in the following sections.

\section{The size of a ball or a sphere with radius $r$ in $S_n$ under the Kendall $\tau$-metric }
\label{sec3}

In this section, we compute the size of a ball or a sphere with radius $r$ in $S_n$ under the Kendall $\tau$-metric and give recursive formulas of $B_K^n(r)$ and $S_K^n(r)$, respectively. Since $B_K^n(r)$ does not depend on the center of the ball, we consider the ball $B_{K}^n(\epsilon_n,r)$ which is a ball with radius $r$ centered at the identity permutation $\epsilon_n$ and denote by $S_{K}^n(\epsilon_n,r)\triangleq\{\sigma\in S_n|d_K(\sigma,\epsilon_n)=w_{k}(\sigma)=r\}$ the sphere centered at $\epsilon_n$ and of radius $r$.

\subsection{The size of a sphere of radius $r$ in $S_n$ under the Kendall $\tau$-metric}
In order to give the property of $S_K^n(r)$, we require some notations and lemmas in \cite{Buzaglo}. For a permutation $\pi=[\pi(1),\pi(2),...,\pi(n)]\in S_n$, the \emph{reverse} of $\pi$ is the permutation $\pi^r\triangleq[\pi(n),\pi(n-1),...,\pi(2),\pi(1)]$. For all $\pi\in S_n$, we have $w_K(\pi)\leq \binom{n}{2}$. For convenience, we denote $S_{K}^n(r)=0$ for $r\geq \binom{n}{2}+1$.

\begin{lemma}\cite[Lemma 1]{Buzaglo} For every $\pi,\epsilon_n\in S_n$,
\begin{equation}
d_{K}(\epsilon_n,\pi)+d_{K}(\epsilon_n,\pi^r)=w_K(\pi)+w_K(\pi^r)=d_{K}(\pi,\pi^r)=\binom{n}{2}\label{eq1}.
\end{equation}
\label{Lm1}
\end{lemma}

By Lemma \ref{Lm1}, we can obtain the following lemma.

\begin{lemma}For any $0\leq i\leq \big\lfloor\frac{\binom{n}{2}}{2}\big\rfloor$,
\begin{equation}
S_{K}^n(i)=S_{K}^n\big(\binom{n}{2}-i\big).\label{eq2}
\end{equation}
\label{Lm2}
\end{lemma}

\begin{proof}
Let $m=\binom{n}{2}$. We just need to prove that $|S_{K}^n(\epsilon_n,i)|=|S_{K}^n(\epsilon_n,m-i)|$. First we define a function $f: S_{K}^n(\epsilon_n,i)\rightarrow S_{K}^n(\epsilon_n,m-i)$, where $f(\pi)=\pi^r$ for any $\pi\in S_{K}^n(\epsilon_n,i)$.

If $\pi\in S_{K}^n(\epsilon_n,i)$, then $w_K(\pi)=i$. By $(\ref{eq1})$, $w_K(\pi^r)=\binom{n}{2}-i=m-i$. Hence, $f(\pi) \in S_{K}^n(\epsilon_n,m-i)$. Moreover, we can easily prove that the function $f$ is reasonable and bijection. Thus,  $S_{K}^n(i)=S_{K}^n\big(\binom{n}{2}-i\big).$
\end{proof}

When $i=0~\text{or}~1$, $S_{K}^n(0)=1$ and $S_{K}^n(1)=n-1$. We will further give a recursive formula of  $S_K^n(r)$ in the following lemma.

\begin{lemma}
\label{Lm3}
For all $4\leq n$ and $2\leq i\leq n-1$,
\begin{equation}
S_{K}^n(i)=\sum_{j=0}^{i}S_{K}^{n-1}(j).\label{eq3}
\end{equation}
Moreover, for all $5\leq n$ and $n\leq i \leq \big\lfloor\frac{\binom{n}{2}}{2}\big\rfloor$,
\begin{equation}
S_{K}^n(i)=\sum_{j=i-(n-1)}^{i}S_{K}^{n-1}(j).\label{eq4}
\end{equation}
\end{lemma}

\begin{proof}
When $4\leq n$ and $2\leq i\leq n-1$, we define $S_{K}^n(\epsilon_n,i,j)\triangleq \{\pi\in S_{K}^n(\epsilon_n,i)|\pi(j)=n\}$ for $n-i\leq j\leq n$, i.e., $\pi\in S_{K}^n(\epsilon_n,i)$ is an element of $S_{K}^n(\epsilon_n,i,j)$ if $n$ appears at the $j$th position of $\pi$. For $\pi \in S_{K}^n(\epsilon_n,i)$, the number of inversions in the permutation $\pi$ is $i$. If $\pi(j)=n$, $(\pi(k),n)$ is an inversion for all $j+1\leq k\leq n$. Hence, for any $\pi \in S_{K}^n(\epsilon_n,i)$, $n$ can only appear at the $j$th position of $\pi$ for every $n-i\leq j\leq n$. So, we obtain that  $S_{K}^n(\epsilon_n,i)=\cup_{j=n-i}^{n}S_{K}^n(\epsilon_n,i,j)$.

For all $n-i\leq j\leq n$, we define $f_j: S_{K}^n(\epsilon_n,i,j) \rightarrow S_{K}^{n-1}(\epsilon_{n-1},i-(n-j))$, where $f_j(\pi)=[\pi(1),\pi(2),...,\pi(j-1),\pi(j+1),...,\pi(n)]$ for any $\pi\in S_{K}^n(\epsilon_n,i,j)$. That is, we delete the element $n$ of $\pi$ to obtain $f_j(\pi)$. Obviously, $f_j$ is injective. For $\pi_1\in S_{K}^{n-1}(\epsilon_{n-1},i-(n-j))$, we define $\pi$ such that $\pi(k)=\pi_1(k)$ for $1\leq k\leq j-1$, $\pi(j)=n$, and $\pi(k)=\pi_1(k-1)$ for $j+1\leq k\leq n$. Then, $\pi\in S_n$ and $w_{K}(\pi)=w_{K}(\pi_1)+(n-j)=i$. Thus, $\pi\in S_{K}^n(\epsilon_n,i,j)$ and $f_j(\pi)=\pi_1$. So, we obtain that $f_j$ is bijection for all $n-i\leq j\leq n$.

Since all the set $S_{K}^n(\epsilon_n,i,j)$ are pairwise disjoint and all the $f_j$ are bijection for all $n-i\leq j\leq n$, we have
\begin{align*}
S_{K}^n(i)&=|S_{K}^n(\epsilon_n,i)|=|\cup_{j=n-i}^{n}S_{K}^n(\epsilon_n,i,j)|=\sum_{j=n-i}^{n}|S_{K}^n(\epsilon_n,i,j)|\\
&=\sum_{j=n-i}^{n}|S_{K}^{n-1}(\epsilon_{n-1},i-(n-j))|=\sum_{j=0}^{i}S_{K}^{n-1}(j).\nonumber
\end{align*}
Similarly, for all $5\leq n$ and $n\leq i \leq \lfloor\frac{\binom{n}{2}}{2}\rfloor$, then $\lfloor\frac{\binom{n}{2}}{2}\rfloor\leq\binom{n-1}{2}$. Thus, for all $i-(n-1)\leq j\leq\lfloor\frac{\binom{n}{2}}{2}\rfloor$, $S_{K}^{n-1}(j)$
exists. So, we also prove that
\begin{equation}
S_{K}^n(i)=\sum_{j=i-(n-1)}^{i}S_{K}^{n-1}(j).\nonumber
\end{equation}
\end{proof}

Furthermore, we give the recursive formula of $S_{K}^n(i)$ for all $4\leq n$ and $4\leq i\leq n-1$ in the following lemma. For convenience, for any function $f(t)$ and two positive integers $i<t$, we denote $\sum_{l=t}^{i}f(l)=0$.
\begin{lemma}
\label{Lm4}
For all $4\leq n$ and $4\leq i\leq n-1$, there exists a unique integer $t$ such that $\binom{t-1}{2}<i\leq\binom{t}{2}$ and $t\geq 4$. Then, we have
\begin{equation}
S_{K}^n(i)=S_{K}^{t}(\binom{t}{2}-i)+\sum_{l=t}^{i-1}\sum_{j=i-l}^{i-1}S_{K}^{l}(j)+\sum_{l=i}^{n-1}\sum_{j=0}^{i-1}S_{K}^{l}(j).\label{eq5}
\end{equation}
\end{lemma}

\begin{proof}
When $4\leq n$ and $4\leq i\leq n-1$, by $(\ref{eq3})$, we have
\begin{equation}
S_{K}^n(i)-S_{K}^{n-1}(i)=\sum_{j=0}^{i-1}S_{K}^{n-1}(j).\label{eq6}
\end{equation}
In $(\ref{eq6})$, we set $n$ to $i+1,...,n$ and obtain $n-i$ equations, respectively. Then by summing all the equations, we have
\begin{equation}
S_{K}^n(i)-S_{K}^{i}(i)=\sum_{l=i}^{n-1}\sum_{j=0}^{i-1}S_{K}^{l}(j).\label{eq7}
\end{equation}

For $j< i$ and $i<n$,  if $S_{K}^n(j)$ and $S_{K}^i(i)$ are known, then by $(\ref{eq7})$ we can compute $S_{K}^n(i)$. In the following, we will compute $S_{K}^i(i)$. By $(\ref{eq4})$, for $i\leq\binom{i-1}{2}$ (i.e., $4\leq i$), we obtain that
\begin{equation}
S_{K}^i(i)-S_{K}^{i-1}(i)=\sum_{j=1}^{i-1}S_{K}^{i-1}(j).\label{eq8}
\end{equation}
For $4\leq i$, we can find an integer $t$ such that $\binom{t-1}{2}<i\leq\binom{t}{2}$ and $t\geq 4$. Then, $\binom{t}{2}+\frac{(t-1)(t-4)}{2}<2i$ and $t<\binom{t-1}{2}$ for $5\leq t$. When $i=4$, we have $t=4$. When $5\leq i$, we have $4\leq t$, $i\leq\binom{t}{2}<2i$, and $t< i$.

Thus, we obtain
\begin{equation}
0\leq\binom{t}{2}-i<i.\label{eq9}
\end{equation}
When $i=4$, $S_{K}^{4}(4)=S_{K}^{4}(\binom{4}{2}-4)=S_{K}^{4}(2)$.

Similarly, when $4<i$, in $(\ref{eq4})$, we set $n$ to $t+1,...,i$ and obtain $i-t$ equations, respectively. By summing all the equations, we have
\begin{equation}
S_{K}^i(i)-S_{K}^{t}(i)=\sum_{l=t}^{i-1}\sum_{j=i-l}^{i-1}S_{K}^{l}(j).\label{eq10}
\end{equation}
Combining $(\ref{eq2})$, $(\ref{eq9})$, and $(\ref{eq10})$, we have
\begin{equation}
S_{K}^i(i)=S_{K}^{t}(\binom{t}{2}-i)+\sum_{l=t}^{i-1}\sum_{j=i-l}^{i-1}S_{K}^{l}(j).\label{eq11}
\end{equation}
When $4\leq i$, we also have $S_{K}^i(i)=S_{K}^{t}(\binom{t}{2}-i)+\sum_{l=t}^{i-1}\sum_{j=i-l}^{i-1}S_{K}^{l}(j)$. When $i=t=4$, the second term (i.e., $\sum_{l=4}^{3}\sum_{j=i-l}^{i-1}S_{K}^{l}(j)$) is zero.
Finally, by $(\ref{eq7})$ and $(\ref{eq11})$, we can obtain the expression of $S_{K}^n(i)$ in the above lemma.
\end{proof}

Specifically, we give the formulas of $S_{K}^n(2)$ and $S_{K}^n(3)$ for all $3\leq n$ as follows.
\begin{lemma}
\label{Lm5}
For all $3\leq n$, we have
\begin{equation}
S_{K}^n(2)=\frac{n(n-1)}{2}-1,\label{eq12}
\end{equation}
\begin{equation}
S_{K}^n(3)=\frac{n^3-7n}{6}.\label{eq13}
\end{equation}
\end{lemma}

\begin{proof}
When $i=2$, by $(\ref{eq6})$, we have
\begin{equation}
S_{K}^n(2)-S_{K}^{2}(2)=\sum_{l=2}^{n-1}\sum_{j=0}^{1}S_{K}^{l}(j).\label{eq14}
\end{equation}
Since $S_{K}^n(0)=1$, $S_{K}^n(1)=n-1$ and $S_{K}^{2}(2)=0$, by $(\ref{eq14})$, we have
\begin{equation}
S_{K}^n(2)=\sum_{l=2}^{n-1}\sum_{j=0}^{1}S_{K}^{l}(j)=\sum_{l=2}^{n-1}l=\frac{n(n-1)}{2}-1.\label{eq15}
\end{equation}

Similarly, when $i=3$, by $(\ref{eq6})$, we have
\begin{equation}
S_{K}^n(3)-S_{K}^{3}(3)=\sum_{l=3}^{n-1}\sum_{j=0}^{2}S_{K}^{l}(j).\label{eq16}
\end{equation}
Since $S_{K}^n(0)=1$, $S_{K}^n(1)=n-1$, $S_{K}^{n}(2)=\frac{n(n-1)}{2}-1$, and $S_{K}^{3}(3)=1$, by $(\ref{eq16})$, we have
\begin{equation}
S_{K}^n(3)=S_{K}^{3}(3)+\sum_{l=3}^{n-1}\sum_{j=0}^{2}S_{K}^{l}(j)=1+\sum_{l=3}^{n-1}\frac{l^2+l-2}{2}=\frac{n^3-7n}{6}.\label{eq17}
\end{equation}

According to $(\ref{eq15})$ and $(\ref{eq17})$, we can obtain the expressions of $S_{K}^n(2)$ and $S_{K}^n(3)$ as $(\ref{eq12})$ and $(\ref{eq13})$, respectively.
\end{proof}

Here, we easily obtain $S_{K}^2(0)=S_{K}^2(1)=1$. By Lemma $\ref{Lm5}$, when $n=3$, we have $S_{K}^3(0)=1$, $S_{K}^3(1)=2$, $S_{K}^3(2)=2$, and $S_{K}^3(3)=1$.  By Lemma $\ref{Lm5}$ and Lemma $\ref{Lm2}$, we have $S_{K}^4(0)=1$, $S_{K}^4(1)=3$, $S_{K}^4(2)=5$, $S_{K}^4(3)=6$, $S_{K}^4(4)=5$, $S_{K}^4(5)=3$, and $S_{K}^4(6)=1$.

If all the $S_{K}^n(j)$ for all $n$ and $j\leq i-1$ are known, by Lemma \ref{Lm4}, we can compute $S_{K}^n(i)$ for $4\leq n$ and $4\leq i\leq n-1$. Next we present an example to compute $S_{K}^n(i)$  in Lemma \ref{Lm4}.
\begin{example}
When $i=4$, $\binom{3}{2}<4\leq\binom{4}{2}$. Then, we obtain $t=4$ in Lemma $\ref{Lm4}$. Furthermore, by $(\ref{eq5})$, we have
\begin{equation}
S_{K}^n(4)=S_{K}^{4}(\binom{4}{2}-4)+\sum_{l=4}^{3}\sum_{j=i-l}^{i-1}S_{K}^{l}(j)+\sum_{l=4}^{n-1}\sum_{j=0}^{3}S_{K}^{l}(j).\nonumber
\end{equation}
By Lemma $\ref{Lm5}$, we have $S_{K}^{4}(\binom{4}{2}-4)=S_{K}^{4}(2)=5$. Thus,
\begin{equation}
S_{K}^n(4)=5+\sum_{l=4}^{n-1}\big(1+(l-1)+\frac{l(l-1)}{2}-1+\frac{l^3-7l}{6}\big)=\frac{n(n+1)(n^2+n-14)}{24}.\label{eq18}
\end{equation}
\end{example}

In the following, we give the recursive formula of $S_{K}^n(i)$ for all $5\leq n$ and $n\leq i\leq \lfloor\frac{\binom{n}{2}}{2}\rfloor$.
\begin{lemma}
\label{Lm6}
For all $5\leq n$ and $n\leq i\leq \lfloor\frac{\binom{n}{2}}{2}\rfloor$, there exists a unique integer $t$ such that $\binom{t-1}{2}<i\leq\binom{t}{2}$ and $t\geq 4$. Then, we have
\begin{equation}
S_{K}^n(i)=S_{K}^{t}(\binom{t}{2}-i)+\sum_{l=t}^{i-1}\sum_{j=i-l}^{i-1}S_{K}^{l}(j)-\sum_{l=n}^{i-1}\sum_{j=i-l}^{i-1}S_{K}^{l}(j).\label{eq19}
\end{equation}
\end{lemma}

\begin{proof}
When $5\leq n$ and $n\leq i\leq \lfloor\frac{\binom{n}{2}}{2}\rfloor$, in $(\ref{eq4})$, we set $n$ to $n+1,...,i$, respectively. Then we obtain $n-i$ equations and sum all the equations. Thus, we have
\begin{equation}
S_{K}^i(i)-S_{K}^{n}(i)=\sum_{l=n}^{i-1}\sum_{j=i-l}^{i-1}S_{K}^{l}(j).\label{eq20}
\end{equation}
By $(\ref{eq11})$ and $(\ref{eq20})$, we have
\begin{equation}
S_{K}^n(i)=S_{K}^{t}(\binom{t}{2}-i)+\sum_{l=t}^{i-1}\sum_{j=i-l}^{i-1}S_{K}^{l}(j)-\sum_{l=n}^{i-1}\sum_{j=i-l}^{i-1}S_{K}^{l}(j).\nonumber
\end{equation}
When $i=n$, the third term (i.e., $\sum_{l=n}^{i-1}\sum_{j=i-l}^{i-1}S_{K}^{l}(j)$) is zero.
\end{proof}

\begin{example}
When $i=5$ and $n=5$, we have $\binom{3}{2}<5\leq\binom{4}{2}$. Then, we obtain $t=4$ in Lemma $\ref{Lm6}$. Furthermore, by $(\ref{eq19})$, we have
\begin{equation}
S_{K}^5(5)=S_{K}^{4}(\binom{4}{2}-5)+\sum_{l=4}^{4}\sum_{j=5-l}^{4}S_{K}^{l}(j).\nonumber
\end{equation}
Thus, $S_{K}^5(5)=S_{K}^{4}(1)+\sum_{j=1}^{4}S_{K}^{4}(j)=3+(3+5+6+5)=22$.
\end{example}

For every $6\leq n$, due to $i=5\leq n-1$, $S_{K}^n(5)$ can be computed by Lemma $\ref{Lm4}$.

Hence, if $S_{K}^n(j)$ are known for all $1\leq j\leq i-1 $ and $n$, we will compute $S_{K}^n(i)$ for all $n$ in the next two steps. For $5\leq i$, there exists a unique integer $t$ such that $\binom{t-1}{2}<i\leq\binom{t}{2}$. Then, for every $2\leq l\leq t-1$, $S_{K}^l(i)=0$. First, when $t\leq l\leq i$, if $i>\lfloor\frac{\binom{l}{2}}{2}\rfloor$, we have $S_{K}^l(i)=S_{K}^l(\binom{l}{2}-i)$ where $\binom{l}{2}-i<i$; otherwise, by Lemma $\ref{Lm6}$, we compute $S_{K}^l(i)$ for $i\leq \lfloor\frac{\binom{l}{2}}{2}\rfloor$. Second, when $i+1\leq l$, we compute $S_{K}^l(i)$ by Lemma $\ref{Lm4}$.

When $i=5$, we can compute $S_{K}^n(5)$ for all $n$. Here, $t=4$. Then, $S_{K}^4(5)=S_{K}^4(\binom{4}{2}-5)=S_{K}^4(1)=3$ and $S_{K}^5(5)=22$ by Lemma $\ref{Lm6}$ in Example $2$. In the following, we will give the formula of $S_{K}^n(5)$ for all $6\leq n$ by Lemma $\ref{Lm4}$.
\begin{example}
When $i=5$ and $6\leq n$, by Lemma $\ref{Lm4}$ and $(\ref{eq7})$, we have
\begin{equation}
S_{K}^n(5)=S_{K}^{5}(5)+\sum_{l=5}^{n-1}\sum_{j=0}^{4}S_{K}^{l}(j).\nonumber
\end{equation}
By Examples $1$ and $2$ and Lemma $\ref{Lm5}$, we have
\begin{align}
S_{K}^n(5)&=22+\sum_{l=5}^{n-1}\big(1+(l-1)+\frac{l(l-1)}{2}-1+\frac{l^3-7l}{6}+\frac{l(l+1)(l^2+l-14)}{24}\big)\nonumber\\
&=\frac{(n-1)(n^4+6n^3-9n^2-74n-120)}{120}\label{eq21}
\end{align}
for all $5\leq n$.
\end{example}

By Lemmas $\ref{Lm2}$, $\ref{Lm4}$, and $\ref{Lm6}$, we can obtain the property of $S_{K}^n(i)$ for all $6\leq i$ and $n$ as follows.
\begin{proposition}
\label{Prp1}
When $6\leq i$, we can compute $S_{K}^n(i)$ for all $5\leq n$ by using Lemmas $\ref{Lm2}$, $\ref{Lm4}$, and $\ref{Lm6}$.
\end{proposition}
\begin{proof}
For all $0\leq i\leq 5$ and $3\leq n$, all the $S_{K}^n(i)$ are computed. We can compute $S_{K}^n(i)$ for all $n$ by using $S_{K}^n(j)$ for all $1\leq j\leq i-1 $ and $n$.

First, we find an integer $t$ such that $\binom{t-1}{2}<i\leq\binom{t}{2}$. For every $t\leq l\leq i$, if $i>\lfloor\frac{\binom{l}{2}}{2}\rfloor$, we have $S_{K}^l(i)=S_{K}^l(\binom{l}{2}-i)$ where $\binom{l}{2}-i<i$; else if $i\leq \lfloor\frac{\binom{l}{2}}{2}\rfloor$, we compute $S_{K}^l(i)$ by Lemma $\ref{Lm6}$. Second, for every $i+1\leq l$, we compute $S_{K}^l(i)$ by Lemma $\ref{Lm4}$. So, we can obtain $S_{K}^n(i)$ for all $5\leq n$ and $6\leq i$.
\end{proof}

\subsection{The size of a ball of radius $r$ in $S_n$ under the Kendall $\tau$-metric}
In this subsection, we will give the size of a ball with radius $r$ in $S_n$ under the Kendall $\tau$-metric and give recursive formula of $B_K^n(r)$ by using $S_K^n(r)$. We easily obtain the following lemma about the relationship between $B_K^n(r)$ and $S_K^n(r)$.

\begin{lemma}
\label{Lm7}
For any $0\leq r\leq \binom{n}{2}$, we have
\begin{equation}
B_{K}^n(r)=\sum_{l=0}^{r}S_{K}^{n}(l).\label{eq22}
\end{equation}
\end{lemma}

Given $S_K^n(i)$ for all $0\leq i\leq r-1$, by Lemmas $\ref{Lm4}$, $\ref{Lm6}$ and $\ref{Lm7}$, we easily obtain the recursion formula of $B_K^n(r)$ in the following theorem.
\begin{theorem}
\label{Thm1}
Suppose $S_K^n(i)$ are known for all $0\leq i\leq r-1$ and $5\leq n$. If $4\leq r\leq \lfloor\frac{\binom{n}{2}}{2}\rfloor$, there exists a unique integer $t$ such that $\binom{t-1}{2}<r\leq\binom{t}{2}$. When $4\leq r \leq n-1$, we have
\begin{equation}
B_{K}^n(r)=\sum_{l=0}^{r-1}S_{K}^{n}(l)+S_{K}^{t}(\binom{t}{2}-r)+\sum_{l=t}^{r-1}\sum_{j=r-l}^{r-1}S_{K}^{l}(j)+\sum_{l=r}^{n-1}\sum_{j=0}^{r-1}S_{K}^{l}(j).\label{eq23}
\end{equation}
When $n\leq r\leq \lfloor\frac{\binom{n}{2}}{2}\rfloor$,  we have
\begin{equation}
B_{K}^n(r)=\sum_{l=0}^{r-1}S_{K}^{n}(l)+S_{K}^{t}(\binom{t}{2}-r)+\sum_{l=t}^{r-1}\sum_{j=r-l}^{r-1}S_{K}^{l}(j)-\sum_{l=n}^{r-1}\sum_{j=r-l}^{r-1}S_{K}^{l}(j).\label{eq24}
\end{equation}
\end{theorem}

Specially, we have $B_{K}^n(0)=1$ and $B_{K}^n(1)=n$. When $r=2$, for all $n\geq 2$, we have
\begin{equation}
B_{K}^n(2)=\sum_{l=0}^{2}S_{K}^{n}(l)=(1+n-1+\frac{n(n-1)}{2}-1)=\frac{(n+2)(n-1)}{2}.\label{eq25}
\end{equation}
When $r=3$, for all $n\geq 3$, we have
\begin{equation}
B_{K}^n(3)=\sum_{l=0}^{3}S_{K}^{n}(l)=(1+n-1+\frac{n(n-1)}{2}-1+\frac{n^3-7n}{6})=\frac{(n+1)(n^2+2n-6)}{6}.\label{eq26}
\end{equation}

\begin{example}
When $r=4$ and $4\leq n$, by Example $1$ and Theorem $\ref{Thm1}$, we have
\begin{align}
B_{K}^n(4)&=\sum_{l=0}^{3}S_{K}^{n}(l)+S_{K}^{4}(\binom{4}{2}-4)+\sum_{l=4}^{3}\sum_{j=4-l}^{4-1}S_{K}^{l}(j)+\sum_{l=4}^{n-1}\sum_{j=0}^{3}S_{K}^{l}(j)\nonumber\\
&=\frac{(n+2)(n+1)(n^2+3n-12)}{24}.\label{eq27}
\end{align}
Moreover, when $r=5$ and $5\leq n$, by Example $3$ and Theorem $\ref{Thm1}$, we have
\begin{align}
B_{K}^n(5)&=\sum_{l=0}^{4}S_{K}^{n}(l)+S_{K}^{n}(5)\nonumber\\
&=\frac{(n+7)n(n^3+3n^2-6n-28)}{120}.\label{eq28}
\end{align}
\end{example}

When $r\geq 6$, we can compute $B_{K}^n(r)$ by using Proposition $\ref{Prp1}$ and Theorem $\ref{Thm1}$.

\section{The nonexistence of a perfect $t$-error-correcting code in $S_n$ under the Kendall $\tau$-metric for some $n$ and $t=2,3,4,~\text{or}~5$ }
\label{sec4}
In this section, we will prove the nonexistence of a perfect $t$-error-correcting code in $S_n$ under the Kendall $\tau$-metric for some $n$ and $t=2,3,4,~\text{or}~5$ by using the sphere-packing upper bound. By Proposition $1$, we give the necessary condition of the existence of a perfect $t$-error-correcting code in $S_n$ under the Kendall $\tau$-metric.
\begin{lemma}
\label{Lm8}
For any $0\leq t\leq \binom{n}{2}$, if there exists one perfect $t$-error-correcting code $C$ in $S_n$ under the Kendall $\tau$-metric. Then, we must have
\begin{equation}
B_{K}^n(t)\cdot |C|=n!.\label{eq29}
\end{equation}
That is, the necessary condition of the existence of a perfect $t$-error-correcting code in $S_n$ under the Kendall $\tau$-metric is $B_{K}^n(t)|n!$.
\end{lemma}
\begin{proof}
By the sphere-packing upper bound in Proposition $1$, if there exists one perfect $t$-error-correcting code $C$ in $S_n$ under the Kendall $\tau$-metric, we must have $B_{K}^n(t)\cdot |C|=n!$. Thus, $B_{K}^n(t)|n!$. So, the necessary condition of the existence of a perfect $t$-error-correcting code in $S_n$ under the Kendall $\tau$-metric is $B_{K}^n(t)|n!$.
\end{proof}

According to Lemma $\ref{Lm8}$, we have the following theorem which illustrate the nonexistence of a perfect $t$-error-correcting code in $S_n$ under the Kendall $\tau$-metric.
\begin{theorem}
\label{Thm2}
For any $0\leq t\leq \binom{n}{2}$, if $B_{K}^n(t)$ has a prime factor $p>n$, then there does not exist one perfect $t$-error-correcting code in $S_n$ under the Kendall $\tau$-metric.
\end{theorem}
\begin{proof}
By Lemma $\ref{Lm8}$, the necessary condition of the existence of a perfect $t$-error-correcting code in $S_n$ under the Kendall $\tau$-metric is $B_{K}^n(t)|n!$. Since $B_{K}^n(t)$ has a prime factor $p>n$, we have $B_{K}^n(t)\nmid n!$. So, we prove the above result.
\end{proof}

In the following, we will dicuss the nonexistence of a perfect $t$-error-correcting code in $S_n$ for some $n$ and $t=2,3,4,~\text{or}~5$ by using Theorem $\ref{Thm2}$.

When $t=2$, by $(\ref{eq25})$, we have $B_{K}^n(2)=\frac{(n+2)(n-1)}{2}$. By Theorem $\ref{Thm2}$, we can prove the nonexistence of a perfect two-error-correcting code in $S_n$, where $n+2>6$ is a prime.

When $t=3$, by $(\ref{eq26})$, we have $B_{K}^n(3)=\frac{(n+1)(n^2+2n-6)}{6}$. First, if $n+1>6$ is a prime, then $B_{K}^n(3)$ have a prime factor $n+1>n$.
Second, we compute $n^2+2n-6$ for $4\leq n\leq 33$ and obtain that $(n+1)(n^2+2n-6)$ has a prime factor $p>n$ except $n=13$ and $n=26$. If $n=13$, $B_{K}^{13}(3)=441=9\times 7^2$. Thus, $441\nmid 13!$. If $n=26$, $B_{K}^{26}(3)=3249=9\times 19^2$. Hence, $3249\nmid 26!$. So, by Theorem $\ref{Thm2}$, we can prove the nonexistence of a perfect three-error-correcting code in $S_n$, where $n+1>6$ is a prime, $n^2+2n-6$ has a prime factor $p>n$, or $4\leq n \leq 33$.

When $t=4$, by $(\ref{eq27})$, we have $B_{K}^n(4)=\frac{(n+1)(n+2)(n^2+3n-12)}{24}$.  First, if $n+1>6$ or $n+2>7$ is a prime, then $B_{K}^n(3)$ have a prime factor $p>n$. Second, we compute $n^2+3n-12$ for $5\leq n\leq 19$ and obtain that $(n ^2+3n-12)(n+1)(n+2)$ has a prime factor $p>n$ except $n=13$. If $n=13$, $B_{K}^{13}(4)=1715=5\times 7^3$. Thus, $1715\nmid 13!$. So, by Theorem $\ref{Thm2}$, we can prove the nonexistence of a perfect four-error-correcting code in $S_n$, where $n+1>6$ or $n+2>7$ is a prime, $n^2+3n-12$ has a prime factor $p>n$, or $5\leq n \leq 19$.

When $t=5$, by $(\ref{eq28})$, $B_{K}^n(5)=\frac{(n+7)n(n^3+3n^2-6n-28)}{120}$. By Theorem $\ref{Thm2}$, we can prove the nonexistence of a perfect five-error-correcting code in $S_n$, where $n+7\geq 12$ is a prime or $n^3+3n^2-6n-28$ has a prime factor $p>n$.

By the above discussion, we have the following theorem.
\begin{theorem}
\label{Thm3}
When $t=2$, there are no perfect two-error-correcting codes in $S_n$, where $n+2>6$ is a prime. When $t=3$, there are no perfect three-error-correcting codes in $S_n$, where $n+1>6$ is a prime, $n^2+2n-6$ has a prime factor $p>n$, or $4\leq n \leq 33$. When $t=4$, there are no perfect four-error-correcting codes in $S_n$, where $n+1>6$ or $n+2>7$ is a prime, $n^2+3n-12$ has a prime factor $p>n$, or $5\leq n \leq 19$. When $t=5$, there are no perfect five-error-correcting codes in $S_n$, where $n+7\geq 12$ is a prime or $n^3+3n^2-6n-28$ has a prime factor $p>n$.
\end{theorem}

\section{Conclusion}
Permutation codes under the Kendall $\tau$-metric have been attracted lots of research interest due to their applications in flash memories. In this paper, we considered the nonexistence of perfect codes under the Kendall $\tau$-metric. We gave the recursive formulas of the size of a ball or a sphere with radius $t$ in $S_n$ under the Kendall $\tau$-metric. Specifically, we gave the polynomial expressions of the size of a ball or a sphere with radius $r$ when $t=2,3,4,~\text{or}~5$. Finally, we used the sphere-packing upper bound to prove that there are no perfect $t$-error-correcting codes in $S_n$ under the Kendall $\tau$-metric  for some $n$ and $t=2,3,4,~\text{or}~5$. Specifically, we proved that there are no perfect two-error-correcting codes in $S_n$, where $n+2>6$ is a prime. We also proved that there are no perfect three-error-correcting codes in $S_n$, where $n+1>6$ is a prime, $n^2+2n-6$ has a prime factor $p>n$, or $4\leq n \leq 33$. We further proved that there are no perfect four-error-correcting codes in $S_n$, where $n+1>6$ or $n+2>7$ is a prime, $n^2+3n-12$ has a prime factor $p>n$, or $5\leq n \leq 19$. We proved that there are no perfect five-error-correcting codes in $S_n$, where $n+7\geq 12$ is a prime or $n^3+3n^2-6n-28$ has a prime factor $p>n$.










\end{document}